\documentclass[12pt]{article}
\usepackage[a4paper]{geometry}
\usepackage[italian,english]{babel}

\usepackage{amsmath}
\usepackage{amsfonts}
\usepackage{amstext}
\usepackage{amssymb}
\usepackage{amsthm}
\usepackage{amscd}

\usepackage[pagebackref,draft=false]{hyperref}
\hypersetup{colorlinks,
linkcolor=myrefcolor,
citecolor=mycitecolor,
urlcolor=myurlcolor}

\usepackage[capitalize]{cleveref}
\usepackage{caption}

\usepackage{xcolor}
\definecolor{myurlcolor}{rgb}{0,0,0.4}
\definecolor{mycitecolor}{rgb}{0,0.5,0}
\definecolor{myrefcolor}{rgb}{0.5,0,0}
\usepackage{graphicx}
\usepackage{tikz}
\usepackage{tikz-cd}
\usepackage{mathrsfs}

\usepackage{etoolbox}
\usepackage{makeidx}
\usepackage{sectsty}
\usepackage{dsfont}
\usepackage{enumitem} 
\usepackage[]{latexsym}
\usepackage{braket}
\usepackage{caption}
\usepackage[utf8]{inputenx}
\usepackage[T1]{fontenc}
\usepackage{lmodern}
\usepackage{textcomp}
\usepackage{microtype}
\usepackage{totcount}
\usepackage{blindtext}

\newtheorem{remark}{Remark}

\newtheorem{proposition}{Proposition}

\newtheorem{example}{Example}

\newtheorem*{proof*}{Proof}
\newtheorem*{note*}{Note}

\newcommand{\be}{\begin{equation}}
\newcommand{\ee}{\end{equation}}
\newcommand{\bea}{\begin{eqnarray}}
\newcommand{\eea}{\end{eqnarray}}
\newcommand{\vsp}{\vspace{0.4cm}}

\newcommand{\stsp}{\mathscr{S}}

\newcommand{\car}{\mathcal{P}}

\title{Differential Geometry of Quantum States, Observables and Evolution}

\author{F. M. Ciaglia$^{1,2,3,6}$\href{https://orcid.org/0000-0002-8987-1181}{\includegraphics[scale=0.7]{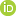}}, A. Ibort$^{4,5,7}$\href{https://orcid.org/0000-0002-0580-5858}{\includegraphics[scale=0.7]{ORCID.png}},  G. Marmo$^{1,2,8}$\href{https://orcid.org/0000-0003-2662-2193}{\includegraphics[scale=0.7]{ORCID.png}} \\
\footnotesize{$^{1}$\textit{ Dipartimento di Fisica ``E. Pancini'', Universit\`a di Napoli Federico II, Napoli, Italy}} \\
\footnotesize{$^{2}$\textit{ INFN-Sezione di Napoli, Napoli, Italy.}} \\
\footnotesize{$^{3}$\textit{ Max Planck Institute for Mathematics in the Sciences, Leipzig, Germany}} \\
\footnotesize{$^{4}$\textit{ ICMAT, Instituto de Ciencias Matem\'{a}ticas (CSIC-UAM-UC3M-UCM)}}  \\
\footnotesize{$^{5}$\textit{ Depto. de Matem\'aticas, Univ. Carlos III de Madrid, Legan\'es, Madrid, Spain}}  \\
\footnotesize{$^{6}$\textit{ e-mail: \texttt{florio.m.ciaglia[at]gmail.com}}, $^{7}$\textit{ e-mail: \texttt{albertoi[at]math.uc3m.es}}} \\
\footnotesize{$^{8}$\textit{ e-mail: \texttt{marmo[at]na.infn.it}}}
}

\date{}

\begin{document}

\maketitle

\abstract{The geometrical description of Quantum Mechanics is reviewed and proposed as an alternative picture to the standard ones.   The basic notions of observables, states, evolution and composition of systems are analised from this perspective, the relevant geometrical structures and their associated algebraic properties are highlighted, and the Qubit example is thoroughly discussed.}

\vsp

\begin{note*}
At the time of the creation and submission of this work to the Lecture Notes of the Unione Matematica Italiana 25, F. M. Ciaglia was not yet affiliated with the Max-Planck-Institut f\"{u}r Mathematik in den Naturwissenschaften in Leipzig with which he is affiliated at the time of the submission of this work to arXiv.
\end{note*}

\section{Introduction}
\label{sec:1}

Finding a unified formalism for both Quantum Mechanics and General Relativity is an outstanding problem facing theoretical physicists.   
From the mathematical point of view, the structural aspects of the two theories could not be more different.  

Quantum Mechanics is prevalently an algebraic theory; the transformation group, in the sense of Klein's programme, is a group of linear transformations (the group of unitary transformations on a Hilbert space for instance). 
General Relativity, on the other hand, sees the triumph of Differential Geometry. 
The covariance group of the theory is the full diffeomorphisms group of space-time.

The usual approach of non-commutative geometry consists on the algebraization of the geometrical background \cite{Co94}; here, we will discuss an opposite attempt:  to geometrise the algebraic description of Quantum Mechanics.
In different terms, we attempt at a description of Quantum Mechanics where non-linear transformations are possible and the full diffeomorphisms group of the carrier space becomes the covariance group of the theory.

Thus we are going to introduce a ``quantum differential manifold'' as a carrier space for our description of Quantum Mechanics, that is, a standard smooth manifold (possibly infinite-dimensional) that will play the role of the carrier space of the quantum systems being studied.  
We shall use a simplifying assumption to avoid introducing infinite dimensional geometry which would go beyond the purposes of this presentation, which is conceptual rather than technical.

Of course, this idea is not new and it has been already explored earlier.   Just to mention a few, we may quote the early attempts by Kibble \cite{Ki79}, the essay by Asthekar and Schilling \cite{Asht99}, the mathematical foundations laid by Cirelli et al \cite{Ci90} and the sistematic search for a geometric picture led by Marmo (see for instance early ideas in the subject in the book \cite{3gbook} and some preliminar results in \cite{Gr05} or the review \cite{Erc10}).  This work is a continuation of this line of thought and contains a more comprehensive description of such attempt.

Let us briefly recall first the various pictures of Quantum Mechanics, emphasising the algebraic structures present in their description.


\subsection{On the many pictures of Quantum Mechanics}

As it is well known,  modern Quantum Mechanics was first formulated by Heisenberg as matrix mechanics,  immediately after Schr\"odinger formulated his wave mechanics.  
These pictures got a better mathematical interpretation by Dirac \cite{Di81} and Jordan \cite{Bo25}, \cite{Jo34} with  the introduction of Hilbert spaces and Transformation Theory.
Further, a sound mathematical formulation was provided by von Neumann \cite{Ne32}.  

In all of these pictures and descriptions,  the principle of analogy with classical mechanics,  as devised by Dirac,  played a fundamental role.
The canonical commutation relations (CCR) were thought to correspond or to be analogous to the Poisson Brackets on phase space.  Within the rigorous formulation of von Neumann,  domain problems were identified showing that at least one among position and momentum operators should be an unbounded operator \cite{Wi47}.    To tackle these problems,  Weyl introduced an ``exponentiated form'' of the commutation relations in terms of unitary operators \cite{We27}, i.e., a projective unitary representation of the symplectic Abelian vector group, interpreted also as a phase-space with a Poisson Bracket.
The $C^*$-algebra of observables,  a generalization of the algebraic structure emerging from Heisenberg picture,  would be obtained as the group-algebra of the Weyl operators.


\subsection{Dirac-Schr\"odinger vs. Heisenberg-Weyl picture}\label{subsec: Dir-Schr vs Hei-Wey}

Even if commonly used, there is not an universal interpretation of the term ``picture'' used above as applied to a particular mathematical emboidement of the axioms used in describing quantum mechanical systems.
The description of any physical system requires the identification of:

\begin{itemize}
\item[i)] \, States.
\item[ii)] \, Observables.
\item[iii) ] \, A probability interpretation.
\item[iv)] \, Evolution.
\item[v)] \, Composition of systems.
\end{itemize}

Thus, in this work, a ``picture'' for a quantum mechanical system will consist of a mathematical description of: i) a collection of states $\mathcal{S}$; ii) a collection of measurable physical quantites or observables $\mathcal{A}$; iii) a statistical interpretation of the theory, that is, a pairing:
\begin{equation}\label{measure}
\mu \colon \mathcal{S} \times \mathcal{A} \to \mathbf{Bo}(\mathbb{R})
\end{equation}  
where $\mathbf{Bo}(\mathbb{R})$ is the set of Borel probability measures on the real line and if $\rho \in \mathcal{S}$ denotes an state of the system and $a$ an observable, then, the pairing $\mu (\rho, a)(\Delta)$ is interpreted as the probability $P(\Delta | a, \rho)$ that the outcome of measuring the observable $a$ lies in the Borelian set $\Delta \subset \mathbb{R}$ if the system is in the state $\rho$.  In addition to these ``kinematical'' framework a ``picture'' of a quantum system should provide iv) a mathematical description of its dynamical behaviour and v) prescription for the composition of two different systems.

\subsubsection{Dirac-Schr\"odinger picture}
Thus, for instance, in the Dirac-Schr\"odinger picture with any physical system we associate a complex separable Hilbert space $\mathcal{H}$.  The (pure) states of the theory are given by rays in the Hilbert space, or equivalently by rank-one orthogonal projectors $\rho =  | \psi \rangle  \langle \psi | /\langle \psi | \psi \rangle$ with $| \psi \rangle \in \mathcal{H}$.     Observables are Hermitean or self-adjoint operators $a$ (bounded or not) on the Hilbert space and the statistical interpretation of the theory is provided by the resolution of the identity $E$ (or spectral measure $E(d\lambda)$) associated to the observable by means of the spetral theorem, $a = \int \lambda E(d\lambda) $.  Thus the probabilty $P(\Delta | a, \rho)$ that the outcome of measuring the observable $A$ when the system is in the state $\rho$ would lie in the Borel set $\Delta \subset \mathbb{R}$, is given by:
\begin{equation}\label{Pdelta}
P(\Delta | a, \rho) = \int_\Delta \mathrm{Tr}( \rho E(d\lambda) ) \, .
\end{equation}   
Moreover the evolution of the system is dictated by a Hamiltonian operator $H$ by means of Schr\"odinger's equation:
$$
i\hbar \frac{d}{dt} | \psi \rangle  = H  | \psi \rangle \, .
$$
Finally, if $\mathcal{H}_A$ and $\mathcal{H}_B$ denote the Hilbert spaces corresponding to two different systems, the composition of them has associated the Hilbert space $\mathcal{H}_A \otimes \mathcal{H}_B$. 

\subsubsection{Heisenberg-Born-Jordan}
In contrast, in the Heisenberg-Born-Jordan picture a unital $C^*$-algebra $\mathcal{A}$ is associated to any physical system.   Observables are real elements $a = a^*$ in $\mathcal{A}$ and states are normalised positive linear functionals $\rho$ on $\mathcal{A}$:
$$
\rho (a^*a) \geq 0 \, , \quad \rho (1_\mathcal{A}) = 1 \, ,
$$
where $1_\mathcal{A}$ denotes the unit of the algebra $\mathcal{A}$.
The GNS construction of a Hilbert space $\mathcal{H}_\rho$, once a state $\rho$ is chosen, reproduces the Dirac-Schr\"odinger picture. Similar statements can be made with respect to the statistical interpretation of the theory.  Given a state $\rho$ and an observable $a \in \mathcal{A}$, the pairing $\mu$ between states and observables, Eq. (\ref{measure}), required to provide a statistical interpretation of the theory is provided by the spectral measure associated to the Hermitean operator $\pi_\rho (a)$ determined by the canonical representation of the $C^*$-algebra $\mathcal{A}$ in the Hilbert space $\mathcal{H}_\rho$ obtained by the GNS construction with the state $\rho$.     Alternatively, given a resolution of the identity, i.e., in the discrete setting,  $E_j \in \mathcal{A}$ such that $E_i\cdot E_j = \delta_{ij} E_j$, and $\sum_j E_j = 1_\mathcal{A}$, we define $p_j(\rho) = \rho (E_j) \geq 0$, $\sum_j p_j(\rho) = 1$.
This provides the probability function of the theory.

The evolution of the theory is defined by means of a Hamiltonian $H \in \mathcal{A}$, $H = H^*$, by means of Heisenberg equation:
$$
i\hbar \frac{da}{dt} = [H, a] \, .
$$
Finally, composition of two systems with $C^*$-algebras $\mathcal{A}_A$ and $\mathcal{A}_B$ would be provided by the tensor product $C^*$-algebra $\mathcal{A}_{AB} = \mathcal{A}_A \otimes \mathcal{A}_B$ (even though there is not a unique completion of the algebraic tensor product of $C^*$-algebras in infinite dimensions, a problem that will not concern us here as the subsequent developments are restricted to the finite-dimensional situation in order to properly use the formalism of differential geometry).

\subsubsection{Other pictures}

The Dirac-Schr\"odinger and Heisenberg-Born-Jordan are far from being the only two pictures of Quantum Mechanics.   Other pictures include the Weyl-Wigner picture, where the an Abelian vector group $V$ with an invariant  symplectic structure $\omega$ is required to possess a projective  unitary representation:
$$
W \colon V \to \mathcal{U}(\mathcal{H}) \, , \quad W(v_1) W(v_2) W(v_1)^\dagger W(v_2)^\dagger = e^{i\omega(v_1,v_2)} \mathbf{1}_{\mathcal{H}} \, .
$$
The tomographic picture \cite{Ib09} has been developed in the past few years and uses a tomographic map $U$ and includes the so called Wigner picture based on the use of pseudo probability distributions on phase space; a picture based on the choice of a family of coherent states has also been partially developed recently (see for instance \cite{Ci17}).  A deep and careful reflexion would be required to analyse the `Lagrangian picture' proposed by Dirac and Schwinger, that would be treated elsewhere. 


\section{A geometric picture of Quantum Mechanics}

As it was discussed in the introduction, the proposal discussed in this work departs from the other ones in setting a geometrical background for the theory, so that the group of natural transformations of the theory becomes the group of diffeomorphisms of a certain carrier space.
In this picture, the carrier space $\car$ we associate with every quantum system is the  Hilbert manifold provided by the complex projective space.
By taking this point of view, states and observables should be defined by means of functions on $\car$.
This carrier space comes equipped with a K\"{a}hlerian structure, i.e., a symplectic structure, a Riemannian structure and a complex structure.
All three tensors, pairwise, satisfy a compatibility condition, two of them will determine the third one.
We will show how to implement on this carrier space the minimalist requirements stated at the beginning of subsection \ref{subsec: Dir-Schr vs Hei-Wey}.

As it was commented before, to properly use the formalism of differential geometry, we shall restrict our considerations to finite dimensional complex projective spaces.
We believe that, at this stage, considering infinite-dimensional systems would introduce a significative amount of technical difficulties without adding any  relevant improving in the exposition of the structural aspects of the ideas we want to convey.
A more thorough analysis of the infinite-dimensional case will be pursued elsewhere.

It is our hope that the ``geometrization'' of Quantum Mechanics can be useful to understand under which conditions any ``generalized'' geometrical quantum theory reduces  to the conventional Dirac-Schr\"odinger picture.

\vsp

{\bfseries The carrier space } $\car$ is taken to be the complex projective space $\mathbb{CP}(\mathcal{H})$ associated with the $n$-dimensional complex Hilbert space $\mathcal{H}$.
This a Hilbert manifold with a K\"ahler structure even in the infinite-dimensional case \cite{Ci90}.
The K\"ahler structure of $\car$ consits of a complex structure $J$, a metric tensor $g$ called the Fubini-Study metric, and a symplectic form $\omega$.
These tensor fields are mutually related according to the following compatibility condition

\begin{equation}
g\left(X\,,J(Y)\right)=\omega\left(X\,,Y\right)\,,
\end{equation}
where $X$ and $Y$ are arbitrary vector fields on $\car$.
The complex sum $h=g + \imath\omega$ is a Hermitian tensor on $\car$.
Following \cite{Asht99, Erc10}, we consider the canonical projection  $\pi\colon \mathcal{H}_{0}\rightarrow\mathbb{CP}(\mathcal{H})\equiv\car$ associating to each non-zero vector $\psi\in\mathcal{H}_{0}$\footnote{$\mathcal{H}_{0}$ denotes the Hilbert space $\mathcal{H}$ with the zero vector removed.} its ray $[\psi]\in\car$ and the Hermitian tensor: 
\begin{equation}
\widetilde{h}=\pi^{*}h=\frac{\langle\mathrm{d}\psi|\mathrm{d}\psi\rangle}{\langle\psi|\psi\rangle} -\frac{\langle\mathrm{d}\psi|\psi\rangle\langle\psi|\mathrm{d}\psi\rangle}{\langle\psi|\psi\rangle^{2}}\,.
\end{equation}
The real part of this tensor is symmetric, and defines the pullback to $\mathcal{H}$ of the Fubini-Study metric $g$, while the imaginary part is antisymmetric and defines the pullback to $\mathcal{H}$ of the symplectic form $\omega$.

We stress that, because our description is tensorial, we may perform any nonlinear transformation without affecting the description of the theory.   
For instance, introducing an orthonormal basis $\{|e_{j}\rangle\}_{j=1,\ldots,n}$ in $\mathcal{H}$, we can write every normalized vector $|\psi\rangle$ in $\mathcal{H}$ as a probability amplitude $|\psi\rangle=\sqrt{\mathrm{p}_{j}}\,\mathrm{e}^{\mathrm{i}\varphi_{j}}\,|e_{j}\rangle$, $\mathrm{p}_{j} \geq 0$ for all $j$.
Clearly, $(\mathrm{p}_{1},...,\mathrm{p}_{n})$ is a probability vector, that is, $\sum\,\mathrm{p}_{j}=1$, while $\mathrm{e}^{\mathrm{i}\varphi_{j}}$ is a phase factor.
Then we can compute $\widetilde{h}$ in this nonlinear coordinate system obtaining:
$$
\widetilde{h}=\frac{1}{4}\left[\langle \mathrm{d}(\ln \vec{\mathrm{p}})\otimes \mathrm{d}(\ln \vec{\mathrm{p}})\rangle_{\vec{\mathrm{p}}} - \langle \mathrm{d}(\ln \vec{\mathrm{p}})\rangle_{\vec{\mathrm{p}}}\otimes \langle\mathrm{d}(\ln \vec{\mathrm{p}})\rangle_{\vec{\mathrm{p}}} \right] +
$$
\begin{equation}
+ \langle \mathrm{d}\vec{\varphi} \otimes\mathrm{d}\vec{\varphi}\rangle_{\vec{\mathrm{p}}} - \langle \mathrm{d}\vec{\varphi}\rangle_{\vec{\mathrm{p}}} \otimes \langle \mathrm{d}\vec{\varphi}\rangle_{\vec{\mathrm{p}}} + \frac{\mathrm{i}}{2}\left[\langle \mathrm{d} \left(\ln \vec{\mathrm{p}}\right)\wedge\mathrm{d}\vec{\varphi}\rangle_{\vec{\mathrm{p}}} - \langle \mathrm{d}\left(\ln \vec{\mathrm{p}}\right)\rangle_{\vec{\mathrm{p}}}\wedge\langle\mathrm{d}\vec{\varphi}\rangle_{\vec{\mathrm{p}}}\right]\,,
\end{equation}
where $\langle\, \cdot\, \rangle_{\vec{\mathrm{p}}}$ denotes the expectation value with respect to the probability vector $\vec{\mathrm{p}}$.
Note that (the pullback of) $g$ is composed of two terms, the first one is equivalent to the Fisher-Rao metric on the space of probability vectors $(\mathrm{p}_{1}\,...\,\mathrm{p}_{n})$, while the second term can be interpreted as a quantum contribution to the Fisher-Rao metric due to the phase of the state \cite{Fac10}.

Given a smooth function $f\in \mathcal{F}(\car)$, we denote by $X_f$, $Y_f$ the vector fields given respectively by: $X_f  = \Lambda (df)$, and $Y_f = R(df)$, where $\Lambda=\omega^{-1}$ and $R=g^{-1}$.   The vector fields $X_f$ will be called Hamiltonian vector fields and $Y_f$, gradient vector fields. 
Note that the compatibility condition among $\omega,g$ and $J$ allows us to write $Y_{f}=J(X_{f})$.

The special unitary group $SU(\mathcal{H})$ acts naturally on $\car=\mathbb{CP}(\mathcal{H})$ by means of isometries of the K\"{a}hler structure.
The infinitesimal version of this action is encoded in a set of Hamiltonian vector fields $\{X_{A} \mid \mathbf{A} \in\mathfrak{su}(\mathcal{H})\}$ such that they close on a realization of the Lie algebra $\mathfrak{su}(\mathcal{H})$ of $SU(\mathcal{H})$.
This means that, given $\mathbf{A},\mathbf{B}\in\mathfrak{su}(\mathcal{H})$, there are Hamiltonian vector fields $X_{A},X_{B}$ on $\car$ such that $[X_{A},X_{B}]=-X_{[A,B]}$  \cite{Abr12}.

The fact that $SU(\mathcal{H})$ acts preserving the K\"{a}hler structure means that the Hamiltonian vector fields for the action preserve $\omega,g$ and $J$, that is, $\mathcal{L}_{X_{A}}\omega=\mathcal{L}_{X_{A}}g=\mathcal{L}_{X_{A}} J=0$ for every $X_{A}$.
Note that this is not true for a Hamiltonian vector field $X_{f}$ associated with a generic smooth function $f$ on $\car$.

It is interesting to note that the Hamiltonian vector fields $X_{A}$ together with the gradient vector fields $Y_{A}=J(X_{A})$ close on a realization of the Lie algebra $\mathfrak{sl}(\mathcal{H})$, that is, the Lie algebra of the complex special linear group $SL(\mathcal{H})$ which is the complexification of $SU( \mathcal{H})$.
In order to see this, we recall the definition of the Nijenhuis tensor $N_{J}$ associated with the complex structure $J$ (see definition $2.10$, and equation $2.4.26$ in \cite{Mor90}):
\begin{equation}
N_{J}(X,Y)=\left(\mathcal{L}_{J(X)}(T)\right)(Y) - \left(J\circ \mathcal{L}_{X}(J)\right)(Y)\,,
\end{equation}
where $X,Y$ are arbitrary vector fields on $\car$.
A fundamental result in the theory of complex manifold is that the $(1,1)$-tensor field $J$ defining the complex structure of a complex manifold must have vanishing Nijenhuis tensor \cite{Nir57}.
This means that the complex structure $J$ on $\car$ is such that $N_{J}=0$, which means:

\begin{equation}
\left(\mathcal{L}_{J(X)}(J)\right)(Y) = \left(J\circ \mathcal{L}_{X}(J)\right)(Y)\,,
\end{equation}
where $X,Y$ are arbitrary vector fields on $\car$.
In particular, if we consider the Hamiltonian vector field $X_{A}$, we know that $\mathcal{L}_{X_{A}}\,J=0$, and thus:

\begin{equation}\label{eqn: lie derivative of complex structure on isospectral states with respect to gradient vector fields}
\left(\mathcal{L}_{J(X_{A})}(J)\right)(Y) = 0
\end{equation}
for every vector field $Y$ on $\car$.
Eventually, we prove the following:

\begin{proposition}\label{prop: SL(H) on homogenous spaces of SU(H)}
Let $\mathbf{A},\mathbf{B}$ be generic elements in the Lie algebra $\mathfrak{su}(\mathcal{H})$ of $SU(\mathcal{H})$
The Hamiltonian and gradient vector fields $X_{A},X_{B},Y_{A},Y_{B}$ on $\car$ close on a realization of the Lie algebra $\mathfrak{sl}(\mathcal{H})$, that is, the following commutation relations among Hamiltonian and gradient vector fields hold:

\begin{equation}
[X_{A}\,,X_{B}]=-X_{[A\,,B]}\,,\;\;\;\;\;[X_{A}\,,Y_{B}]=- Y_{[A\,,B]}\,,\;\;\;\;\;[Y_{A}\,,Y_{B}]=X_{[A\,,B]}\,.
\end{equation}
\end{proposition}

\begin{proof}
The first commutator follows directly from the fact that there is a left action of $SU(\mathcal{H})$ on $\car$ of which the Hamiltonian vector fields $X_{A}$ are the fundamental vector fields.
Regarding the second commutator, we recall that $Y_{\mathbf{A}}=J(X_{\mathbf{A}})$ and  that $\mathcal{L}_{X_{\mathbf{A}}}\,J=0$, so that:

\begin{equation}
\begin{split}
[X_{\mathbf{A}}\,,Y_{\mathbf{B}}]=&\mathcal{L}_{X_{\mathbf{A}}}\,\left(J(X_{\mathbf{B}})\right)= \\
=&\left(\mathcal{L}_{X_{\mathbf{A}}}\,J\right)(X_{\mathbf{B}}) + J\left(\mathcal{L}_{X_{\mathbf{A}}}\, X_{\mathbf{B}}\right)=\\
=& J\left([X_{\mathbf{A}}\,, X_{\mathbf{B}}]\right)=- Y_{ [\mathbf{A}\,,\mathbf{B}]}
\end{split}
\end{equation}
as claimed.
Finally, using equation \eqref{eqn: lie derivative of complex structure on isospectral states with respect to gradient vector fields} together with the fact that $J\circ J=-\mathrm{Id}$ because it is a complex structure, we obtain:

\begin{equation}
\begin{split}
[Y_{\mathbf{A}}\,,Y_{\mathbf{B}}]=& \mathcal{L}_{J(X_{\mathbf{A}})}\,\left(J(X_{\mathbf{B}})\right)=\\
=& \left(\mathcal{L}_{J(X_{\mathbf{A}})}(J)\right)X_{\mathbf{B}} + J\left(\mathcal{L}_{J(X_{\mathbf{A}})}X_{\mathbf{B}}\right)=\\
=&J\left([Y_{\mathbf{A}}\,,X_{\mathbf{B}}]\right)=X_{[\mathbf{A}\,,\mathbf{B}]}
\end{split}
\end{equation}
as claimed. \qed
\end{proof}

Since $\car$ is a compact manifold, all vector fields are complete, in particular, the Hamiltonian and gradient vector fields of Prop. \ref{prop: SL(H) on homogenous spaces of SU(H)} are complete.
This means that the realization of the Lie algebra $\mathfrak{sl}(\mathcal{H})$ integrates to an action of $SL(\mathcal{H})$ on $\car$.
We will see that this action on $\car$ allows us to define an action of $SL(\mathcal{H})$ on the space $\stsp$ of quantum states.

\begin{remark}
Instead of the complex projective space, we may as well have started with a generic homogeneous space of $SU(\mathcal{H})$ as a carrier manifold.
Every such manifold is a compact K\"{a}hler manifold, and the Hamiltonian and gradient vector fields associated with elements in $\mathfrak{su}(\mathcal{H})$  close on a realization of the Lie algebra of $SL(\mathcal{H})$ which integrates to a group action.
Indeed, all we need to prove an analogue of Prop. \ref{prop: SL(H) on homogenous spaces of SU(H)} is a K\"{a}hler manifold on which $SU(\mathcal{H})$ acts by means of isometries of the K\"{a}hler structure.

The complex projective space may be selected requiring the holomorphic sectional curvature of $\car$ to be constant and positive.
Indeed, from the  Hawley-Isuga Theorem \cite{Ha53}, \cite{Ig54}, it follows that complex projective spaces are the only (connected and complete) K\"ahler manifolds of constant and positive holomorphic sectional curvature (in our setting equal to $2/\hbar$) up to K\"{a}hler isomorphisms.
\end{remark}

\vsp
{\bfseries Observables } are real functions $f\in\mathcal{F}(\car)$ satisfying:
\begin{equation}\label{kahlerian}
\mathcal{L}_{X_f} R = 0 \, ,
\end{equation}
i.e., such that the Hamiltonian vector fields defined by them are isometries for the symmetric tensor $R=g^{-1}$. 
In particular, if $F$ is a complex-valued function on $\car$ generating a complex-valued Hamiltonian vector field $X_F$ which is Killing for $g$ (hence for $R$), then  there necessarily exist $a$, $b$ Hermitean operators such that \cite{Asht99, Ci90, Erc10, Sk02}:

\begin{equation}
F ([\psi]) = \frac{\langle \psi | a | \psi\rangle}{\langle\psi|\psi\rangle} + \imath  \frac{\langle \psi | b | \psi\rangle}{\langle\psi|\psi\rangle}\,.
\end{equation}
This result is interesting but not unexpected, Hamiltonian vector fields are infinitesimal generators of symplectic transformations.  If they also preserve the  Euclidean metric, they must be infinitesimal generators of rotations, then the intersection of symplectic and rotations are unitary transformations, whose infinitesimal generators are (skew) Hermitean matrices. 
From what we have just seen it follows that the {\bfseries observables } can be identified with the expectation-value functions:

\begin{equation}
e_{a} ([\psi]) = \frac{\langle \psi | a | \psi\rangle}{\langle\psi|\psi\rangle} 
\end{equation}
with $a$ a Hermitian operator on $\mathcal{H}$ (notice that, consistently, $A = \imath\,  a$ is an element in the Lie algebra $\mathfrak{su}(\mathcal{H})$ of the unitary group $SU(\mathcal{H})$). 
We will denote the family of observables as $\mathcal{K}(\car)$ or simply $\mathcal{K}$ for short.

We find out that, under adequate conditions, the family of functions $\mathcal{K}$  constitutes a Lie-Jordan algebra.  
Indeed, the space of K\"ahlerian functions, that is, those satisfying condition (\ref{kahlerian}) above, because of Hawley-Igusa theorem carries a natural $C^*$-algebra structure and its real part a Lie-Jordan one (\cite{Ha53}, \cite{Ig54}, \cite{Bo47}, \cite[Thm. 7.9]{Ko69}).
By using a GNS construction for the $C^*$-algebra we get a Hilbert space, returning to the Dirac-Schr\"odinger picture.

\vsp

By using $\Lambda$ ($\omega$) and $R$ ($g$) we can define the following brackets among functions on $\car$:
\begin{equation}
\{f_{1},f_{2}\}:=\Lambda(\mathrm{d}f_{2},\mathrm{d}f_{1})=\omega(X_{f_{1}},X_{f_{2}})=X_{f_{2}}(f_{1})\,,
\end{equation}
\begin{equation}
(f_{1},f_{2}):=R(\mathrm{d}f_{1},\mathrm{d}f_{2})=g(Y_{f_{1}},Y_{f_{2}})\,.
\end{equation}
The antisymmetric bracket $\{\cdot,\cdot\}$ is a Poisson bracket since it is defined starting from a symplectic form.
Furthermore, being $[X_{f_{1}},X_{f_{2}}]=-X_{\{f_{1},f_{2}\}}$ for every smooth functions $f_{1},f_{2}$ on $\car$, and since $[X_{A},X_{B}]=-X_{[A,B]}$ for the Hamiltonian vector fields associated with $A,B\in\mathfrak{su}(\mathcal{H})$, we have:
\begin{equation}\label{XAB}
-X_{[A,B]}=[X_{A},X_{B}]=[X_{f_{a}},X_{f_{b}}]=-X_{\{f_{a},f_{b}\}}\,,
\end{equation}
where we have switched the notation $e_a$ to $f_a$ to make formulas more familiar and readable.
From (\ref{XAB}) it follows:
\begin{equation}
\{f_{a},f_{b}\}=f_{\imath[a,b]}\,,
\end{equation}
where we used the fact that $A=\imath a$ and $B=\imath b$.
This means that $(\mathcal{K}(\car),\{\cdot,\cdot\})$ is a Lie algebra.

On the other hand,  a direct computation \cite{Erc10} shows that: 
\begin{equation}
(f_{a},f_{b}):=R(\mathrm{d}f_{a},\mathrm{d}f_{b})=g(Y_{a},Y_{b})=f_{a\odot b} -f_{a}\cdot f_{b}\,,
\end{equation}
where $a\odot b=ab + ba$.
Then, we may define the symmetric bracket:
\begin{equation}
<f_{1},f_{2}>:=(f_{1},f_{2}) + f_{1}\cdot f_{2}
\end{equation}
so that on the subspace of {\bfseries observables} we have:
\begin{equation}
<f_{a},f_{b}>=f_{a\odot b}\,.
\end{equation}
Because of the properties of the symmetric product $\odot$ on Hermitean operators, the bracket $<\cdot,\cdot>$ turns out to be a Jordan product.
Furthermore, the set of {\bfseries observables} endowed with the antisymmetric product $\{\cdot,\cdot\}$ and the symmetric product $<\cdot,\cdot>$ is a Lie-Jordan algebra \cite{Ci17b, Ci17c, Fa12}.
By complexification, that is, considering complex-valued functions $F_{A}=f_{a_{1}} + \imath f_{a_{2}}$ for some Hermitean $a_{1},a_{2}$, we obtain a realization of the $C^{*}$-algebra $\mathcal{B}(\mathcal{H})$ by means of smooth functions on $\car=\mathbb{CP}(\mathcal{H})$ according to \cite{Ci90, Erc10}:
\begin{equation}
\begin{split}
F_{A}\star F_{B} &:= \frac{1}{2}\left(F_{A}\cdot F_{B} +  (F_{A},F_{B}) + \imath \{F_{A},F_{B}\}\right)=\\
&= \frac{1}{2}\left(<F_{A},F_{B}> + \imath \{F_{A},F_{B}\}\right)=F_{AB}\,.
\end{split}
\end{equation}
We may extend this product to arbitrary complex-valued functions obtaining a $\star$-product.

\vsp

Because we are in finite dimensions we can consider the critical points of the observables (that is, expectation value functions).  
An observable is said to be generic if all critical points are isolated.
The values of the observable function at these critical points constitute the spectrum of the observable.
The set of critical point of a generic observable may be thought of as the geometrical version adapted to $\car\equiv\mathbb{CP}(\mathcal{H})$ of an orthogonal resolution of the identity on $\mathcal{H}$.
If a critical point is not isolated, the critical set is actually a submanifold   of (real) even dimension.
If the observable has value zero in some critical set, this set is a complex projective space.

We postpone a complete discussion of the critical values of a given observable and restict our analysis to generic observables.
With the help of any generic observable we can now define {\bfseries quantum states}.
The space $\stsp$ of {\bfseries quantum states} is a subset of $\mathcal{K}$ whose elements are defined as follows.
A function in $\mathcal{K}$ will define a state if its evaluation on the set of isolated critical points of any generic observable will be a probability distribution on $n$-elements, i.e., a discrete probability distribution.
In a certain sense, we may think of quantum states (in finite dimensions) as a sort of noncommutative generalization of discrete probability distributions.
Essentially, {\bfseries quantum states} are identified with  the expectation-value functions
$$
\mathrm{e}_{\rho}([\psi])=\frac{\langle\psi|\rho|\psi\rangle}{\langle\psi|\psi\rangle}
$$
associated with density operators, that is, $\rho\in\mathcal{B}(\mathcal{H})$, $\rho=\rho^{\dagger}$, $\langle\psi|\rho|\psi\rangle\geq0$ for all $|\psi\rangle\in\mathcal{H}$,  and $\mathrm{Tr}(\rho)=1$.
In the infinite-dimensional case $\rho$ must be trace-class in order for this last requirement to make sense.

On the other hand, the expectation value function associated with a quantum state will define a ``continuous'' probability distribution on the carrier space provided by the complex projective space.
Essentially, a quantum state is identified with an observable (expectation value function) $\mathrm{e}_{\rho}\in\mathcal{K}$ such that $\mathrm{e}_{\rho}([\psi])\geq0$ for all $[\psi]\in\car$ ($\rho\in\mathcal{B}(\mathcal{H})$ is a positive semidefinite operator), and ($\mathrm{Tr\,}\rho =1$):
\begin{equation}
\int_{\car}\mathrm{e}_{\rho}\,\mathrm{d}\nu_{\omega}=1\,,
\end{equation}
where $\mathrm{d}\nu_{\omega}=\omega^{n}$ is the symplectic volume form normalized by $\int_{\car}\mathrm{d}\nu_{\omega}=1$.
This point of view would be closer to the point of view taken by Gelfand and Naimark to define states as functions  of positive-type in the group algebra of any Lie group.  
They would be of positive-type when pulled back to the group from the homogeneous space.
It is clear that they form a convex body whose extremals are the pure quantum states.

In this context, the pairing map between quantum states and observables given by:
\begin{equation}
\mathrm{E}(\mathrm{e}_{\rho},f_{a})=\int_{\car}\,f_{a}\,\mathrm{e}_{\rho}\,\mathrm{d}\nu_{\omega}
\end{equation}
is interpreted as the mean value for the outcome of a measurement of the observable $f_{a}$ on the quantum state $\mathrm{e}_{\rho}$.

\begin{remark}
In the infinite-dimensional case we must pay attention to topological and measure-theoretical issues since {\bfseries quantum states} are required to be measurable with respecto to the symplectic measure $\nu_{\omega}$, while {\bfseries observables} are not.
\end{remark}

We may define the following map:

\begin{equation}
\car\ni\,[\psi]\mapsto \rho_{\psi}:=\frac{|\psi\rangle\langle\psi|}{\langle\psi|\psi\rangle}\in\mathcal{B}(\mathcal{H})\,.
\end{equation}
This map allows us to identify the points of $\car$ with rank-one projectors on $\mathcal{H}$, and, since rank-one projector are density operators, we identify the points in the carrier space $\car$ with particular quantum states.
These quantum states are precisely the extremal points of the convex set $\stsp$ of all quantum states, that is, pure quantum states.
In this context, the expectation value function $e_{\rho_{\Psi}}$ associated with the pure quantum state $\rho_{\Psi}$ encodes the transition probabilities between the normalized vector $|\tilde{\Psi}\rangle$ associated with $|\Psi\rangle$ and every other normalized vector $|\tilde{\Phi}\rangle$ in $\mathcal{H}$:
\begin{equation}
\mathrm{e}_{\rho_{\Psi}}([\Phi])=\frac{\langle\Phi|\rho_{\Psi}|\Phi\rangle}{\langle\Phi|\Phi\rangle}=\frac{\langle\Phi|\Psi\rangle\langle\Psi|\Phi\rangle}{\langle\Phi|\Phi\rangle \langle\Psi|\Psi\rangle}=|\langle\tilde{\Psi}|\tilde{\Phi}\rangle|^{2} \, .
\end{equation}

\vsp

Recalling that a quantum state is a positive function on $\car$, that is, $\mathrm{e}_{\rho}\geq 0$,  we can define the rank of a {\bfseries quantum state} as the codimension of the closed submanifold $\mathrm{e}_{\rho}^{-1}(0)\subset\car$.
With this definition, it is clear that the rank is  invariant under the group of diffeomorphisms.
As a matter of fact it is possible to show that the complex special Lie group $SL(\mathcal{H})$ acting on $\car$ by means of diffeomorphisms acts transitively on the space of states with the same rank, providing in this way a stratification of the space of states.
To be able to change the rank of a state, to describe decoherence,we need to use semigroups.

Writing $|\psi_{G}\rangle\equiv G|\psi\rangle$ with $G\in SL(\mathcal{H})$, the action of the special linear group $SL (\mathcal{H})$ on the carrier space $\car$ reads:
\begin{equation}
[G]\colon [\psi]\mapsto [G]([\psi])= [\psi_{G}]\,.
\end{equation}
In terms of the rank-one projector $\rho_{\psi}$ we have:
\begin{equation}
\rho_{\psi}\mapsto G\cdot\rho_{\psi}= \frac{G|\psi\rangle\langle\psi|G^{\dagger}}{\langle\psi|G^{\dagger}\,G|\psi\rangle}=\frac{G^{\dagger}\rho_{\psi} G}{\mathrm{Tr}(G^{\dagger}\rho_{\psi} G)}\,.
\end{equation}
We may generalize this action to any density operator by setting:
\begin{equation}\label{Grho}
G\cdot\rho=\frac{G^{\dagger}\rho G}{\mathrm{Tr}(G^{\dagger}\rho G)}\,.
\end{equation}
However, because the action is nonlinear this is an assumption that can not derived from the action on rank-one projectors.
By means of this action we would get an orbit of density operators and thus an orbit of probability distributions once we identify the density operators with their associated expectation-value functions.  
Each orbit being characterised by the rank of $\rho$.   For a system with $n$ levels ($\dim \mathcal{H} = n$) we would get $n$ different orbits.    
The one of maximal dimension would be the bulk, while the boundary of the closed convex body $\stsp$ of quantum states would be the union of orbits of dimensions less than $n$.   
The geometry of $\stsp$ as developed in \cite{Ci17b, Ci17c, Gr05} will be exposed in Sect. \ref{subsec: Quantum states and open systems}.

\vsp

The {\bfseries statistical interpretation} of the theory is provided by a geometric measure.   
The idea is to extend the notion of spectral measure to a geometric manifold as it was proposed for instance by Skulimowski \cite{Sk02} in the case of the complex projective space $\car$.  
Thus we may use a slightly extended notion defined as: a geometric positive-operator-valued measure (GPOV-measure) on a space of states of a geometric quantum theory is a map $p \colon \mathcal{B}(\mathbb{R}) \to \mathcal{K}(\car)$ (where $\mathcal{B}(\mathbb{R})$ denotes the $\sigma$-algebra of Borelian sets in $\mathbb{R}$) such that:

\begin{enumerate}
\item Positivity monotonicity and normalization: $$0 \leq p(\emptyset)([\psi]) \leq p(\Delta)([\psi]) \leq p (\mathbb{R}) ([\psi]) = 1\, .$$

\item Additivity: $\mu$ is additive, i.e., 
$$
p(\cup_{k= 1}^n \Delta_k) ([\psi]) = \sum_{k= 1}^n p(\Delta_k)  ([\psi])\, ,
$$ 
for all $[\psi]\in\car$, $n \in \mathbb{N}$, $\Delta_k$, $k = 1, \ldots, n$, disjoint Borel sets on $\mathbb{R}$.
\end{enumerate}

Thus, consider for instance a GPOV-measure $p$ with finite support, $\mathrm{supp\,} p = \{ \lambda_1, \ldots, \lambda_r \}$, then the statistical interpretation of the theory will be provided, as in the standard pictures, by the probability distribution $\mathrm{p}_k ([\psi]) = p(\{ \lambda_k\}) ([\psi]) \geq 0$, $\sum_{k} \mathrm{p}_k ([\psi]) = 1$.  

In general a GPOV-measure $p$ will be provided by any observable $e_a$ by means fo the corresponding spectral measure $E_a(d\lambda)$ associated to the Hermitean operator $a$, that
is 
$$
p(\Delta)([\psi]) = \int_\Delta \mathrm{Tr\, }(\rho_\psi)E(d\lambda)) = \int_\Delta \frac{\langle \psi |E(d\lambda) | \psi \rangle }{\langle \psi | \psi \rangle} \, ,
$$
in accordance with the probabilistic interpretation of a physical theory, Eq. (\ref{measure}), and the standard pictures, Eq. (\ref{Pdelta}).
\vsp

 {\bfseries Hamiltonian evolution}, or evolution of closed systems, will be defined by the Hamiltonian vector field $X_h$ associated with the observable $h$, that is:
$$ 
\frac{df}{dt} = X_h(f) \, .
$$
We call the observable $h$ the Hamiltonian function for the evolution.

 {\bfseries The composition of systems } will be discussed in Section \ref{composition}.

\subsection{Quantum states and open systems}\label{subsec: Quantum states and open systems}

The geometry of $\stsp$  as a closed convex body in the affine ambient space $\mathfrak{T}_{1}$ of Hermitean operators on $\mathcal{H}$ with trace equal to $1$ has been extensively developed in \cite{Ci17b, Ci17c}.
In these works, it is shown that there exist two bivector fields $\Lambda$ and $\mathcal{R}$ on $\mathfrak{T}_{1}$ by means of which the infinitesimal version of the action of $SL(\mathcal{H})$ on $\stsp$ may be recovered in terms of Hamiltonian and gradient-like vector fields.
In this case, the Poisson bivector field $\Lambda$ does not come from a symplectic structure, and the symmetric bivector field $\mathcal{R}$ is not invertible (there is no metric tensor $g=\mathcal{R}^{-1}$).

\subsubsection{The Qubit}

We will briefly recall here the results of \cite{Ci17b, Ci17c} concerning the geometry of the space of all states, pure or mixed for the qubit.  
Every 2 by 2 Hermitean matrix $A$ may be written in the form:
$$
A = \left[ \begin{array}{cc}  x_0 + x_3 & x_1 - ix_2 \\ x_1 + ix_2 & x_0 - x_3 \end{array}  \right] \, ,
$$
or, written as combination of Pauli matrices:
$$
\sigma_0 = \left[ \begin{array}{cc}  1 & 0 \\ 0 & 1 \end{array}  \right] \, ,\quad
\sigma_1 = \left[ \begin{array}{cc}  0 & 1 \\ 1 & 0 \end{array}  \right] \, ,\quad
\sigma_2 = \left[ \begin{array}{cc}  0 & -i \\ i & 0 \end{array}  \right] \, ,\quad
\sigma_3 = \left[ \begin{array}{cc}  1 & 0 \\ 0 & -1 \end{array}  \right] \, ,
$$
we get:
$$
A =  a_0 \sigma_0 +  a_1 \sigma_1 + a_2 \sigma_2 + a_3 \sigma_3 \, .
$$
In particular it is well known that any density operator $\rho$, that is $\mathrm{Tr\, } \rho = 1$, $0 \leq \rho^2 \leq \rho$ can be written as: 
$$
\rho = \frac{1}{2} (\sigma_0 + \mathbf{x}\cdot \boldsymbol{\sigma}) \, , \quad || \mathbf{x} || \leq 1 \, .
$$
Thus the space $\mathcal{S}$ of all qubit states is the Bloch's ball in $\mathbb{R}^3$:
$$
\mathcal{S} = \{ \mathbf{x} \in \mathbb{R}^3 \mid x_1^2 + x_2^2 + x_3^2 \leq 1 \} \, .
$$

\begin{remark}
In the $n$-dimensional case $\mathcal{H}\cong\mathbb{C}^{n}$, this construction allows us to identify pure states as rank-one projectors in $\mathcal{B}(\mathcal{H})$.
However, they will only be a closed portion of the $2(n-1)$-dimensional unit sphere in $\mathbb{R}^{2n -1}$.
\end{remark}

The tensor field $\Lambda$ in the coordinates $x_1,x_2,x_3$ reads:

\begin{equation}\label{Lambda}
\Lambda = \epsilon_{ijk} x_i \frac{\partial}{\partial x_j} \wedge  \frac{\partial}{\partial x_k} \, ,
\end{equation}
while the symmetric tensor field $\mathcal{R}$ is given by:
$$
\mathcal{R} = \delta_{jk} \frac{\partial}{\partial x_j} \otimes \frac{\partial}{\partial x_k} - x_jx_k  \frac{\partial}{\partial x_j} \otimes \frac{\partial}{\partial x_k} \, .
$$

\begin{remark}[On the bivector $\Lambda$]

The choice of the bivector $\Lambda$ requires a comment.
If we identify $\mathbb{R}^3$ with the dual of the Lie algebra of $SU(2)$, we can consider $\mathbf{x} = (x_1, x_2, x_3)$ as the linear functions on the dual of the Lie algebra $\mathfrak{su}(2)$ of $SU(2)$.  Therefore the Lie bracket of $\mathfrak{su}(2)$ induces a Poisson bracket on $\mathfrak{su}(2)^*$ whose Poisson tensor is given by $\Lambda$.   Notice that $SU(2)$ is the subgroup of unitary operators of determinant one of the group of unitary operators of $\mathcal{H} = \mathbb{C}^2$.    

An alternative way of deriving $\Lambda$ is to consider the projection map $S^3 \to S^2$ related with the momentum map associated with the symplectic action of the unitary group on the Hilbert space $\mathcal{H}$.   Such map $\mu \colon \mathcal{H} \to \mathfrak{su}(2)^*$ provides a symplectic realization of the Poisson manifold $\mathfrak{su}(2)^*$.  
\end{remark}

In this context, observables correspond to affine functions on $\stsp$, that is, $f_{a} = a^j x_j + a_{0}$, $a_{0},a_j \in \mathbb{R}$.
Consequently, the Hamiltonian vector fields $X_a = \Lambda (df_{a}, \cdot )$, and the gradient-like vector fields $Y_{a} = \mathcal{R}(df_{a}, \cdot )$ are given by:
$$
X_{a} = \epsilon_{jkl} a^j x_k \frac{\partial}{\partial x_l} \, , \qquad Y_{a} = a^j  \frac{\partial}{\partial x_j} - a^k x_k \Delta \, ,
$$
with $\Delta = x_j \partial /\partial x_j$ the dilation vector field on $\mathbb{R}^3$.
Lie algebra generated by  the family of vector fields $X_f$, $Y_f$ is the Lie algebra $SL(2, \mathbb{C})$.

It is now possible to construct a Lie-Jordan algebra (see for instance \cite{Ci17b, Ci17c, Fa12}) with commutative Jordan product $\circ$ and Lie product $\{ \cdot, \cdot \}$ on the space of observables (affine functions) out of the tensors $\mathcal{R}$ and $\Lambda$.  Such algebra is defined by:
$$
x_j \circ x_k = \mathcal{R}(dx_j, dx_k) + x_jx_k \, , \quad \{ x_j, x_k\} = \Lambda(dx_j, dx_k) \, .
$$
Then we find:
$$
x_j \circ x_j = 1 \, , \quad x_j \circ x_k = 0 \, , \quad \forall j\neq k \, .
$$
Combining the Jordan product and the Lie product we can define:
$$
x_j \star x_k = x_j \circ x_k + i \{x_j, x_k \}
$$
and we get:
$$
x_j \circ x_k = \frac{1}{2}(x_j \star x_k + x_k \star x_j) \, , \quad \{ x_j, x_k\} = -\frac{i}{2}(x_j \star x_k - x_k \star x_j) 
$$

The involution * will be complex conjugation and we get a $C^*$-algebra which can be used either to go back to the Hilbert space via de GNS construction or to go back to the Heisenberg picture  if we realise the algebra in terms of operators.

Let us remark that as our algebras are described by means of tensor fields, it is evident that the particular coordinate system we use to describe the ball does not play any role.   The convexity structure may well become hidden.
For instance, parametrising Bloch's ball with spherical coorodinates $(r,\theta, \varphi)$, the relevant tensor fields would be:
 $$
\mathcal{R} = (1- r^2) \frac{\partial}{\partial r}\otimes \frac{\partial}{\partial r} + \frac{1}{r^2}\frac{\partial}{\partial \theta} \otimes \frac{\partial}{\partial \theta} + \frac{1}{r^2 \sin^2\theta} \frac{\partial}{\partial \varphi} \otimes \frac{\partial}{\partial \varphi} \, ,
$$
and
$$
\Lambda = \frac{1}{r \sin \theta} \frac{\partial}{\partial \theta} \wedge \frac{\partial}{\partial \varphi} \, .
$$ 
It is now clear by inspection that Hamiltonian vector fields and gradient vector fields are tangent to the sphere of pure states $S^2 = \{r = 1 \}$. 
The interior of the ball is an orbit of the group  $SL(2, \mathbb{C})$ and it is generated by the functions $r\cos \theta$, $r \sin \theta \sin \varphi$ and $r \sin \theta \cos \varphi$ by means of $\mathcal{R}$ and $\Lambda$.

To describe decoherence one needs vector fields which are generators of semigroups so that they will be directed vector fields not vanishing on the sphere of  pure states.

\subsubsection{Open quantum systems: the GKLS equation}

Let us consider the Kossakowski-Lindblad equation (see for instance \cite{Ci17b} and references therein):
$$
\frac{d}{dt} \rho = L(\rho) \, ,
$$
with initial data $\rho(0) = \rho_0$ and,
\begin{eqnarray*}
 L(\rho) &=& - i [H, \rho] + \frac{1}{2} \sum_j ([V_j\rho, V_j^\dagger] + [V_j, \rho V_j^\dagger]) \\
 &=& - i [H, \rho] -  \frac{1}{2} \sum_j [V_j^\dagger V_j, \rho]_+ + \sum_j V_j \rho V_j^\dagger \, ,
\end{eqnarray*}
say with, $\mathrm{Tr\,\,}V_j = 0$, and $\mathrm{Tr\,}(V_j^\dagger V_k) = 0$ if $j\neq k$.
We see immediately that the equations of motion split into three terms:
 \begin{enumerate}
\item Hamiltonian term:  $  - i [H, \rho] $
 
\item Symmetric term (or gradient) : $  -  \frac{1}{2} \sum_j [V_j^\dagger V_j, \rho]_+$

\item Kraus term (or jump vector field):  $ \sum_j V_j \rho V_j^\dagger$.

\end{enumerate}
 
 It is possible to associate a vector field with this equation of motion \cite{Ci17c, Ci17b}.   
It turns out that the one associated with the Kraus term $Z$, is a nonlinear vector field, similar to the nonlinear vector field $Y$, associated with the symmetric tensor, the gradient vector field.
The nonlinearity pops up because the two maps are not trace preserving therefore we have to introduce a denominator for the map to transform states into states.
The ``miracle'' of the Kossakowski-Lindblad form of the equation is that the two nonlinearities cancel each other so that the resulting vector field is actually linear \cite{Ci17c, Ci17b}.
 
\begin{example}[The phase-damping of a q-bit]
Consider now:
 $$
 L(\rho) = - \gamma (\rho - \sigma_3 \rho \sigma_3) \, ,
 $$
 we find the vector field:
 $$
 Z_L = - 2\gamma \left( x_1 \frac{\partial}{\partial x_1} + x_2 \frac{\partial}{\partial x_2} \right) 
 $$
 which allows to visualise immediately the evolution.
\end{example}

\section{Composition of systems}\label{composition}

As we mentioned in the introductory remarks, the composition of two systems $A,B$ in the Dirac-Schr\"odinger picture is simply the tensor product $\mathcal{H}_A \otimes \mathcal{H}_B = \mathcal{H}_{AB}$.    If our starting input is the complex projective space $P(\mathcal{H})$, we cannot consider the Cartesian product $P(\mathcal{H}_A)\times P(\mathcal{H}_B)$ because this would not contain all the information of the composite system, it would not contain what Schr\"odinger called the principal characteristic of quantum mechanics: the entangled states.
According to our general procedure, we should associated with the composite system  the complex projective space related to $\mathcal{H}_{A}\otimes\mathcal{H}_{B}$.
It is easy to visualise the situation in the case of the qubit.   Here the complex projective space is $S^2$, for two qubits we would have $S^2 \times S^2$.  However if we take correctly the tensor product $\mathbb{C}^2 \otimes \mathbb{C}^2$ and then the associated complex projective space, we would get $P(\mathbb{C}^2 \otimes \mathbb{C}^2) = \mathbb{CP}^3$ which is six-dimensional and not four-dimensional as $S^2 \times S^2$.    The additional states account for the entangled states, while the immmersion of $S^2 \times S^2$ into $\mathbb{CP}^3$ would give the space of separable states.

A more intrinsic way would be to consider  the tensor product $\mathcal{A}_A \otimes \mathcal{A}_B = \mathcal{A}_{AB}$ of the $C^{*}$-algebras $\mathcal{A}_A$ and $\mathcal{A}_B$ of expectation value functions on the K\"{a}hler manifolds of the physical subsystems, use the GNS construction to build a Hilbert space on which the chosen completion $\overline{\mathcal{A}_A \otimes \mathcal{A}}_B$ would have an irreducible representation,  and the associated complex projective space should be considered to represent the composition of the two systems. 
Having the space describing  the composite system we could proceed as usual.


\subsection{Decomposing a system}

Given the $C^*$-algebra $\mathcal{A}_{AB}$ of the total system we may now look for the two $C^*$-algebras, say $\mathcal{A}_A$ and $\mathcal{A}_B$, of the original components as subalgebras of the total $C^*$-algebra.   We would ask of the subalgebras that they have in common only the identity and they commute with each other.  Moreover we require that $\mathcal{A}_A\otimes \mathcal{A}_B$, after completion, be isomorphic with the total algebra.   

To recover the states of the two subsystems we may define two projections, say:  $\pi_A \colon \mathcal{S}_{AB} \to \mathcal{S}_A$, $\pi_A(\rho) = \rho_A$, $\rho_A(a) = \rho (a \otimes 1_B)$, and  $\pi_B \colon \mathcal{S}_{AB} \to \mathcal{S}_B$, $\pi_B(\rho) = \rho_B$, $\rho_B(b) = \rho (1_A \otimes b)$, for all $a\in \mathcal{A}_A$, $b \in \mathcal{A}_B$, $\rho \in \mathcal{S}_{AB}$.

We find that $\rho_{AB} \neq \rho_A \otimes \rho_B$.   Indeed, the quantity $\mathrm{Tr\,} (\rho_{AB} - \rho_A \otimes \rho_B)^k$ for every $k$, say integer, would provide possible measures of entanglement.  

As a matter of fact both $\rho_A$ and $\rho_B$ are no more elements of the complex projective space associated to the two subsystems.  
They turn out to be, by construction, non-negative, Hermitian and normalised linear functionals, each one for the total $C^*$-algebra, that is, they are  mixed states.

If we consider a unitary evolution on the composite system, say $U \rho U^\dagger$, we could consider, for any trajectory $U(t) \rho_0 U(t)^\dagger$, the projection on the subsystem $\mathcal{A}$, say:
$$
\rho_A(t) (a) = (U(t) \rho_0 U(t)^\dagger) (a \otimes 1_B) = 
\rho_0 (U(t)^\dagger (a \otimes 1_B) U(t)) \, .
$$
If $\rho_0$ is a separable pure state, it will project onto a pure state onto the subsystem.  However, as time goes by, $\rho (t)$ will not be separable anymore and we get an evolution of a mixed state for the subsystem out of the evolution of a pure state for the total system.    By letting the separable state $\rho_0$ vary by changing the second factor in $\mathcal{A}_B$ while preserving the first factor in $\mathcal{A}_A$, we would get an evolution for the projection on the system $\mathcal{A}_A$ which originates from the same initial point but would evolve with different trajectories, each one depending on the second factor.    

When is it possible to describe the projected evolution by means of a vector field?  This means that the projected trajectories would be described by a semigroup because the evolution would change the rank.   The answer to this question was provided by A. Kossakowski and further formalised by Gorini, Kossakowski, Sudarshan and Linbland \cite{Go76}, \cite{Li76}.   The trajectories would be solutions of the Kossakowski-Lindblad master equation.


\section{Conclusions and discussion}

The geometric description of mechanical systems based on the K\"ahler geometry of the space of pure states of a closed quantum system is proposed as an alternative picture of Quantum Mechanics.
The composition of systems is also briefly discussed in this setting.

The tensorial description of Quantum Mechanics would allow for generic nonlinear transformations, hopefully more flexible to deal with nonlinearities, like entanglement, entropies and so on.  Thus, the geometrical-tensorial description allows to recover as a covariance group of our description the full diffeomorphism group (similarly to General Relativity).

To illustrate the various aspects of the theory we study finite-dimensional systems, with a particular focus on the qubit example.  
It is shown that in the carrier space of the theory there are Hamiltonian and gradient vector fields $X_a$ and $Y_b$ generating the action of the Lie group $SL(\mathcal{H})$.
This action may be extended to the closed convex body $\stsp$ of all quantum states.
From the point of view of the affine ambient space $\mathfrak{T}_{1}$ of Hermitean operators with trace equal to $1$ in which $\stsp$ naturally sits, we find that this action has, again, an infinitesimal description in terms of Hamiltonian and gradient-like vector fields closing on a realization of the Lie algebra $\mathfrak{sl}(\mathcal{H})$.
Moreover, from the perspective of the evolution, to describe semigroups we have to introduce Kraus vector fields on $\mathfrak{T}_{1}$.
Having described the dynamics in terms of vector fields will provide a framework to describe non-Markovian dynamics.  
States in the ``bulk'' may have as ``initial conditions'' pure, extremal states.  The evolution would be described by a family of semigroups associated with
higher order vector fields.


%
\section*{Acknowledgements}
The authors acknowledge financial support from the Spanish Ministry of Economy and Competitiveness, through the Severo Ochoa Programme for Centres of Excellence in RD (SEV-2015/0554).
AI would like to thank partial support provided by the MINECO research project MTM2014-54692-P and QUITEMAD+, S2013/ICE-2801.   GM would like to thank the support provided by the Santander/UC3M Excellence Chair Programme.

\end{document}